%% file: int_freqsel.tex
\documentclass[10pt, conference]{IEEEtran}

\usepackage{amssymb}
\usepackage{amsmath}
\usepackage{epic,eepic}
\usepackage{epsfig}
\usepackage{cite}
\usepackage{verbatim}
\usepackage{enumerate}
\usepackage{subfigure}
\usepackage{multicol}
\usepackage{picinpar}
\usepackage{pstricks}
\usepackage{pst-node}

\newtheorem{theorem}{\bf Theorem}

\begin{document}
\setcounter{page}{1}
\title{Multiple Access Outerbounds and the Inseparability of Parallel Interference Channels}
\author{\authorblockN{Viveck R. Cadambe, Syed A. Jafar}
\authorblockA{Electrical Engineering and Computer Science\\
University of California Irvine, \\
Irvine, California, 92697, USA\\
Email: {vcadambe@uci.edu, syed@uci.edu}\\ \vspace{-1cm}}}

\maketitle
\vspace{12pt}
\begin{abstract} 
It is known that the capacity of parallel (multi-carrier) Gaussian point-to-point, multiple access and broadcast channels can be achieved by separate encoding for each subchannel (carrier) subject to a power allocation across carriers. In this paper we show that such a separation does not apply to parallel Gaussian interference channels in general. A counter-example is provided in the form of a $3$ user interference channel where separate encoding can only achieve a sum capacity of $\log(\mbox{SNR})+o(\log(\mbox{SNR}))$ per carrier while the actual capacity, achieved only by joint-encoding across carriers, is $3/2\log(\mbox{SNR}))+o(\log(\mbox{SNR}))$ per carrier. As a byproduct of our analysis, we propose a class of multiple-access-outerbounds on the capacity of the $3$ user interference channel.
\end{abstract}

\section{Introduction}
The study of parallel Gaussian channels is motivated by the frequency-selective or time-varying nature of the wireless channel. With multi-carrier modulation, (assuming no inter-carrier interference (ICI)) a frequency selective channel can be viewed as a set of parallel channels with channel coefficients that vary from one carrier to another but may be assumed constant (flat-fading) over each carrier. Similarly, if inter-symbol-interference (ISI) is absent, the time-varying channel gives rise to parallel channels whose values are fixed during each symbol but vary from one symbol to another. In this paper, we will use the terminology of frequency-selective channels and multi-carrier modulation to refer to parallel Gaussian channels. It is understood that the model is equally applicable to the time-varying channel as well. 

It is well known that over the parallel Gaussian point-to-point channel, coding separately over the individual subchannels (carriers)  achieves the capacity subject to optimal power allocation. Thus the capacity of the parallel Gaussian point-to-point channel is equal to the sum of the capacities of the point-to-point Gaussian subchannels with corresponding powers chosen through the water-filling algorithm. Similarly, it has also been shown that separate coding over each carrier is optimal for parallel Gaussian multiple access (MAC) and broadcast (BC) channels \cite{dtse:mac,tse:parallelbc}. The separability of parallel Gaussian point-to-point, MAC and BC is useful because it provides a direct connection between the single-carrier channel models studied extensively in classical information theory and the frequency-selective (or time varying) channels that may be more relevant in practice. Coding schemes designed for the classical (single carrier) models can be applied directly to multi-carrier systems subject to a power allocation across carriers. A key question that remains open is whether such a separation holds for other Gaussian networks, and in particular, if separate encoding is optimal for multi-carrier interference networks.

Much work on multi-carrier interference networks (e.g. in the context of DSL \cite{chung_cioffi:parallel_strongint, scutari_palomar_etal:wfillint,leung_etal:wfillint, pang_scutari_etal:parallel_int, yu_cioffi:parallel_int, yu_ginis_cioffi:parallel_int,Cendrillon_etal, Yu_Lui,taek_chung:thesis}) has focused on optimal power allocation across carriers under the assumption of \emph{separate coding} over each carrier. For the two user parallel interference channel with strong interference it is shown in \cite{taek_chung:thesis} that indeed the sum capacity is the sum of the rates that can be achieved by separately encoding over each carrier subject to an overall power optimization. For the case where more than $2$ users are present or when the channels are not restricted to the strong interference case, since the capacity of even the single-carrier interference channel is not known, usually the rate optimization is carried out under the practically motivated assumption that all interference is to be treated as noise. Both centralized and distributed algorithms, some of which are based on game-theoretic formulations, have been proposed for this ``dynamic spectrum management'' problem and the optimality and convergence properties of these algorithms have been established under the separate encoding assumption. 

Joint encoding of multiple-carriers has been used recently in \cite{cadambe_jafar:Kuserint} to characterize the sum capacity per carrier, of the $K$ user multi-carrier Gaussian interference channel. The sum capacity (per carrier) is found to be  $$ C(\mbox{SNR}) = \frac{K}{2}\log(\mbox{SNR}) + o(\log(\mbox{SNR})),$$ 
where SNR represents the signal to noise power ratio. In other words, the $K$ user interference channel has $K/2$ degrees of freedom\footnote{Also known as multiplexing-gain  (See \cite{tse_zheng:divmultiplexing}) or capacity pre-log.} per orthogonal time and frequency dimension. The key to the capacity characterization is the idea of interference alignment (see \cite{jafar_shamai:dofx} and the references therein) - a construction of signals such that they cast overlapping shadows at the receivers where they constitute interference while they remain distinghishable at the receivers where they are desired. The interference alignment constructions proposed in \cite{cadambe_jafar:Kuserint} are based on \emph{joint encoding} over multiple frequencies. Due to interference alignment, the joint encoding scheme of \cite{cadambe_jafar:Kuserint} outperforms the dynamic spectrum management schemes of \cite{ scutari_palomar_etal:wfillint,leung_etal:wfillint, pang_scutari_etal:parallel_int, yu_cioffi:parallel_int, yu_ginis_cioffi:parallel_int} in terms of degrees of freedom \footnote{Interestingly, in both cases interference is treated as noise, so no multiuser detection is involved.}. However, it has not been shown that this joint encoding is \emph{necessary} to achieve capacity. Interestingly, another recent work in \cite{cadambe_jafar_shamai:intconst} has provided examples where interference alignment is achieved over a single-carrier interference channel, i.e., with separate encoding. Thus, it remains unclear whether the capacity of multi-carrier interference channels can be achieved by separate encoding over each carrier and a power allocation across carriers. It is this open problem that we address in this paper.

The main result of this paper is that unlike the point-to-point, multiple-access and broadcast channels, in general separate coding \emph{does not suffice} to achieve the capacity of the interference channel. We establish this result by constructing a counterexample - a $3$-user frequency-selective interference channel where separate coding can only achieve a sum rate of  $\log(\mbox{SNR})+o(\log(\mbox{SNR}))$ per carrier while the capacity is shown to be   $3/2\log(\mbox{SNR})+o(\log(\mbox{SNR}))$ per carrier. Thus, parallel interference channels are, in general, inseparable. 

As a byproduct of our analysis we also propose a class of outerbounds on the capacity of the $3$ user interference channel. These outerbounds share the property that one receiver (possibly aided by a genie and/or noise reduction) is able to decode all messages - so that the multiple-access channel capacity to the genie-aided receiver becomes an outerbound on the sum capacity of the $3$ user interference channel. The MAC outerbounds can be viewed as a generalization of Carliel's outerbound \cite{carliel:interference} on the $2$ user interference channel to the case of more than $2$ users. These outerbounds play an important role in identifying singularity conditions for interference channels that do not achieve the $K/2$ degrees of freedom. However, the bounds are generally loose in the degrees of freedom sense and tighter bounds at high SNR may be obtained by an application of Carlieal's outerbound on each of the $2$ user channels contained within the $K$ user interference channel.

We start with the classical (single-carrier) Gaussian $3$ user interference channel.

\section{The Gaussian $3$ User Interference channel}
\begin{figure}[!tbp]
\begin{center}\input{3user.eepic}\end{center}
\caption{The $3$ user interference channel}
\label{fig:3user}
\end{figure}
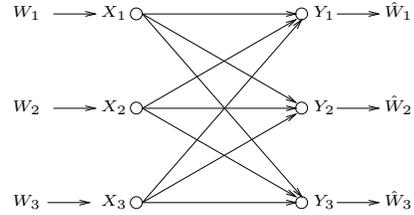

We study the $3$ user (single-carrier) Gaussian interference channel whose input-output relations are described as follows 
$$Y_{i}(n) = \sum_{j=1}^{3} h_{i,j} X_{j}(n) + Z_{i}(n), i=1,2,3 $$
where at the $n$th symbol, $Y_i(n)$ and $Z_i(n)$ respectively represent the received signal and the noise at the $i$th receiver, and $X_j(n)$ represents the signal transmitted by the $j$th transmitter. $h_{i,j}$ represents the channel gain between transmitter $j$ and receiver $i$. All channel gains are assumed to be non-zero and known to all the nodes in the network. Transmitter $i$ has message $W_i$ for receiver $i$ for $i=1,2,3$. The noise $Z_i(n)$ is a zero-mean additive white Gaussian noise (AWGN), assumed independent identically distributed (i.i.d.) across users and symbols. With the noise power at each receiver normalized to unity, the total transmit power can be expressed as $E\left[ \displaystyle \frac{1}{N}\sum_{i=1}^3\sum_{n=1}^{N} |X_i(n)|^2 \right] \leq \mbox{SNR}$, where $N$ is the length of the codeword. The rate of the $i$th user is defined as $R_i(\mbox{SNR}) = \frac{\log(|W_i|)}{N}$ where $|W_i|$ is the cardinality of the message set corresponding to message $W_{i}$. A rate vector $\mathbf{R}(\mbox{SNR})=(R_1(\mbox{SNR}), R_2(\mbox{SNR}), R_3(\mbox{SNR}))$ is said to be \emph{achievable} if messages $W_i,i=1,2,3,$ can be encoded at rates $R_i(\mbox{SNR}), i=1,2,3$ so that the probability of decoding error can be made arbitrarily small by choosing an appropriately large $N$. The capacity region $C(\mbox{SNR})$ represents the set of all achievable rate vectors in the network. The sum capacity $C_{\Sigma}(\mbox{SNR})$ of the network is defined as 
$$ C_{\Sigma}(\mbox{SNR}) = \max_{\mathbf{R}(\textrm{SNR}) \in C(\textrm{SNR})} \sum_{i=1}^{3} R_i(\mbox{SNR})$$
	The number of degrees of freedom of the network is defined as
$$ d_{\Sigma} = \lim_{\textrm{SNR} \to \infty} \frac{C_{\Sigma}(\mbox{SNR})}{\log(\mbox{SNR})}$$
Equivalently, $d_{\Sigma}$ is the total number of degrees of freedom of the network if and only if we can write
$$ C_{\Sigma}(\mbox{SNR}) = d_{\Sigma} \log(\mbox{SNR}) + o(\log(\mbox{SNR})).$$

\begin{theorem}
\label{thm:main}
Consider the $3$ user interference channel where 
$$\frac{h_{i,j}}{h_{i,i}} = \frac{h_{k,j}}{h_{k,i}} $$ 
for some $i,j,k \in \{1,2,3\}, j \neq k, k \neq i, i \neq j$.
Then, this interference channel has $1$ degree of freedom, or equivalently, the sum capacity of the interference channel may be expressed as 
$$ C_{\Sigma}(\mbox{SNR}) = \log(\mbox{SNR}) + o(\log(\mbox{SNR}))$$
\end{theorem}
\begin{proof}

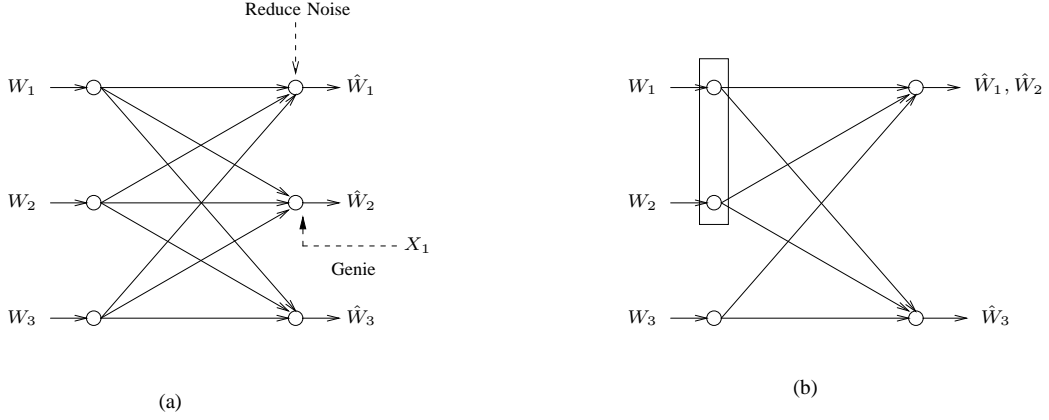
\begin{figure*}[!tbp]
\begin{center}\input{3userconverse.eepic}\end{center}
\caption{The converse argument of Theorem \ref{thm:main}}
\label{fig:3usrconverse}
\end{figure*}

Achievability is trivial since setting $W_2=W_3=\phi$, we get a point-to-point Gaussian channel whose capacity is known to be of the form $\log(\mbox{SNR})+o(\log(\mbox{SNR}))$. 
We show the converse for the special case where $k=1, i=2, j=3$. i.e., we consider the case where 
$$\frac{h_{2,3}}{h_{2,2}} = \frac{h_{1,3}}{h_{1,2}}  = \gamma, \gamma \neq 0. $$ 
By symmetry, the converse extends to all other cases. Consider any achievable coding scheme. Let a genie give $X_1$ to receiver $2$ (Figure \ref{fig:3usrconverse}(a)). Now, receiver $2$ can cancel the interference from transmitter $1$ to obtain $\tilde{Y}_2$ which may be written as 
\begin{eqnarray} \tilde{Y}_2 &=& h_{2,2} X_2 + h_{2,3} X_3 + Z_2 \nonumber\\
\tilde{Y}_2 &=& h_{2,2} ( X_2 +  \gamma X_3) + Z_2 \label{eqn:Rx2}\end{eqnarray}
The dependence on the symbol index $n$ is dropped above for convenience.
Note that any achievable scheme over the original channel is also achievable over this genie-aided channel and therefore, the genie does not affect the converse argument (See for example \cite{jafar_shamai:dofx}). Now, since we started  with an achievable coding scheme, receiver $1$ can decode $X_1$ reliably and therefore, cancel the effect of $X_1$ from $Y_1$ to obtain 
\begin{eqnarray} \tilde{Y}_1 &=& h_{1,2} X_2 + h_{1,3} X_3 + Z_1 \nonumber\\
 \tilde{Y}_1 &=& h_{1,2} (X_2 + \gamma X_3) + Z_1 \label{eqn:Rx1} \end{eqnarray}
 Note that receiver $2$ is able to decode $W_2$ from $\tilde{Y}_2$. Equations (\ref{eqn:Rx1}) and (\ref{eqn:Rx2}) imply that by reducing the variance of $Z_1$ sufficiently, we can ensure that $\tilde{Y}_2$ is a noisy version of $\tilde{Y}_1$. Therefore, in a channel with sufficiently reduced noise, we can ensure that receiver $1$ can decode $W_2$ as well. Note that reducing noise can only increase the capacity of a channel and therefore the converse argument is not affected. 
Thus, by reducing noise and with the aid of a genie (Figure \ref{fig:3usrconverse}(a)), we have ensured that any message which can be decoded at receiver $2$ can be decoded at receiver $1$ as well. Now, in this channel, we can let transmitters $1$ and $2$ co-operate to form a MIMO two user interference channel as in Figure \ref{fig:3usrconverse}(b). Again, note that allowing transmitters to co-operate cannot reduce capacity. Thus, the  MIMO interference channel of Figure \ref{fig:3usrconverse}(b) has a capacity region that contains the capacity region of the $3$ user interference channel of Figure \ref{fig:3user}. Reference \cite{jafar_fakhereddin:MIMOint} has shown that the MIMO interference channel of Figure \ref{fig:3usrconverse}(b) has $1$ degree of freedom meaning that its capacity is of the form $\log(\mbox{SNR}) + o(\log(\mbox{SNR}))$.  Therefore, we have shown that 
$$ C_{\Sigma}(\mbox{SNR}) = \log(\mbox{SNR}) + o(\log(\mbox{SNR}))$$ 
and the converse argument is complete.
\end{proof}
\section{The parallel Gaussian $3$ user interference channel}
\label{sec:parallel_gaussian}
The parallel Gaussian interference channel consisting of $M$ parallel subchannels may be expressed as 
\begin{eqnarray*}\mathbf{Y}_{i}(n) &=& \displaystyle\sum_{j=1}^{3} \mathbf{H}_{i,j} \mathbf{X}_{j}(n) + \mathbf{Z}_{i}(n), i=1,2,3 \end{eqnarray*}
where, corresponding to the $n$th symbol $\mathbf{Y}_i(n), \mathbf{Z}_i(n), \mathbf{X}_j(n)$ are $M \times 1$ vectors whose $M$ entries represent the signal received at receiver $i$ over the $M$ sub-channels, the i.i.d. AWGN experienced by receiver $i$ over the $M$ carriers, and the signal transmitted by the $j$th transmitter over the $M$ carriers, respectively. $\mathbf{H}_{i,j}$ is a $M \times M$ diagonal matrix whose $m$th diagonal entry represents the channel gain between transmitter $j$ and receiver $i$ corresponding to the $m$th subchannel. All channel gains are assumed to be non-zero and known apriori to all nodes. Messages, achievable rates, power constraints, capacity and degrees of freedom are defined in the usual manner as described in the previous section.

Let $C_{\Sigma}^{[m]}(\mbox{SNR})$ denote the sum capacity of the interference channel over the $m^{th}$ carrier and $\mbox{SNR}_m$ denote the total transmit power constraint over the $m^{th}$ carrier. The main question addressed in this correspondence is the following - Can the capacity (per carrier) of the parallel interference channel be expressed as the sum of the capacities achieved by the constituent interference channels over each carrier, i.e.,
\begin{eqnarray} 
C_{\Sigma}(\mbox{SNR}) &=& \frac{1}{M}\sum_{m=1}^{M} C_{\Sigma}^{[m]}(\mbox{SNR}_m) \label{eqn:powalloc2}
\end{eqnarray}
for some power allocation vector $(\mbox{SNR}_1,\mbox{SNR}_2, \ldots \mbox{SNR}_M)$ such that 
\begin{eqnarray}
\sum_{m=1}^{M} \mbox{SNR}_{m} \leq \mbox{SNR} \label{eqn:powalloc1}. \label{eqn:powalloc22}
\end{eqnarray}

The existence of a power vector satisfying the above equations would imply that a capacity-optimal scheme is to code separately over each carrier with power $\mbox{SNR}_m$ allocated to the $m$th carrier. We will use the result of Theorem \ref{thm:main} to construct a parallel interference channel where independent coding over its subchannels is suboptimal. Specifically, we construct a multi-carrier interference channel where,
\begin{enumerate}
\item Interference alignment achieves 3/2 degrees of freedom so that the capacity of the channel is $3/2 \log(\mbox{SNR})) + o(\log(\mbox{SNR}))$ per carrier.
\item Each subchannel has only $1$ degree of freedom meaning that separate encoding over each carrier is suboptimal since it can only achieve a capacity of $\log(\mbox{SNR})+o(\log(\mbox{SNR}))$ per carrier.
\end{enumerate}
This is easily done as follows. Consider the case where we have $2$ carriers, so $M=2$. Let
\begin{eqnarray}
\mathbf{H}_{i,j} &=& \left[\begin{array}{cc} 1 & 0 \\ 0 & 1\end{array}\right], \forall i \neq j, i,j \in \{1,2,3\} \label{eqn:parallel1}\\
\mathbf{H}_{2,2} &=& \mathbf{H}_{1,1} = \left[\begin{array}{cc}1 & 0 \\ 0 & -1\end{array}\right]\\
\mathbf{H}_{3,3} &=& \left[\begin{array}{cc}-1 & 0 \\ 0 & 1\end{array}\right] \label{eqn:parallel2}
\end{eqnarray}
It can be easily verified that each subchannel in this above channel satisfies the conditions of Theorem \ref{thm:main} so that each subchannel has $1$ degree of freedom. Furthermore, it can also be verified that by beamforming messages along vector $[ 1~~1]^{T}$ at each user ensures that at all receivers, interference aligns along $[1~~1]^T$ . The desired messages can be decoded along the zero-forcing vector $[1 -1]^T$ at each receiver and thus $3/2$ degrees of freedom are achievable over this network. 
\begin{figure}[!tbp]
\begin{center}\resizebox{3.5in}{2.7in}{\includegraphics*{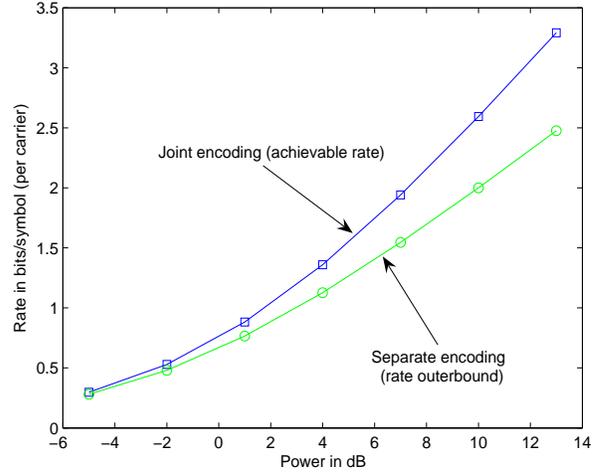}}\end{center}
\caption{A comparison of the performance of joint coding versus separate coding on the parallel $3$ user interference channel}
\label{fig:jtvssep}
\end{figure}
We now state this result formally in a theorem.
\begin{theorem}
Parallel Gaussian interference channels are in general, not separable. Equivalently, in general there do not exist coding schemes such that
\begin{eqnarray*} 
C_{\Sigma}(\mbox{SNR}) &=& \frac{1}{M}\sum_{m=1}^{M} C_{\Sigma}^{[m]}(\mbox{SNR}_m) \\
\sum_{m=1}^{M} \mbox{SNR}_{m} &\leq& \mbox{SNR}  
\end{eqnarray*}
\end{theorem}
The above theorem clearly implies the sub-optimality of separate coding over each carrier of the $3$ user interference channel in general.

Figure \ref{fig:jtvssep} illustrates the suboptimality of separate coding over each carrier in comparison with the interference alignment based joint coding scheme for the channel described in equations (\ref{eqn:parallel1})-(\ref{eqn:parallel2}). The outerbound for the rate achievable by separate encoding (plotted in Figure \ref{fig:jtvssep}) is derived later in Section \ref{sec:macouterbound} (See example $1$ in the referred section). Note that in Figure \ref{fig:jtvssep} the separate encoding outerbound is \emph{not} limited to schemes that treat interference as noise. The outerbound is for the sum of the Shannon capacities of the interference channels over each carrier, and thus allows arbitrary encoding/decoding schemes, possibly including multi-user detection, with the only restriction that independent data is sent through separate codebooks over separate carriers. Thus, separate encoding schemes that treat interference as noise may perform much worse than the separate encoding outerbound. The \emph{joint encoding} achievable rate in Figure \ref{fig:jtvssep} on the other hand, is based on treating interference as noise and is only an innerbound on the rates achievable with joint encoding. Thus, even the simple joint encoding scheme which uses Gaussian codebooks (not known to be optimal) and treats interference as noise is able to achieve higher rates than could be achieved with the best separate encoding schemes. Further, while our counterexample is based on a degrees of freedom argument which is meaningful only at high SNR, the plot in Figure \ref{fig:jtvssep} shows that the capacity with joint encoding can be substantially higher than with separate encoding even at moderate to low SNRs. Lastly, note that no such example can be constructed for the parallel Gaussian point to point, multiple access and broadcast channels because in all those cases separate coding over each carrier is capacity-optimal for \emph{any} channel realization.

An interesting interpretation of the counterexample presented above is the following. Consider a game that is played between two players. The players will pick the channel coefficient values for a (single-carrier) $3$ user interference channel. Player $1$ intends to maximize the number of degrees of freedom of the channel. Player $2$ wants to minimize the number of degrees of freedom of the channel. In this game, player $1$ moves first and player $2$ moves second. During his turn, player $1$ is allowed to select the values of all the channel coefficients. Player $2$ can only change the value of $1$ channel coefficient after the values have been chosen by player $1$. Which channel coefficient player $2$ is allowed to change is also decided by player $1$. There is a constraint that all channel co-efficients (diagonal terms of the channel matrix) must be non-zero. First, consider the constant interference channel. Note that \cite{cadambe_jafar_shamai:intconst} has already shown that there exist $3$ user channels with close to $3/2$ degrees of freedom. Therefore, in absence of player $2$, player $1$ can design a channel that will achieve close to $3/2$ degrees of freedom. However, if player $2$ can control any one of the channel co-efficients, he can use the result of Theorem \ref{thm:main} to win the game by reducing the number of degrees of freedom to unity. For example, if player $2$ has control of $h_{1,2}$, he can choose the channel co-efficient to be equal to $\frac{h_{1,3}h_{2,2}}{h_{2,3}}$ to ensure that the channel has only $1$ degree of freedom. Thus, in a constant single-carrier channel, player $2$ wins the game.

Now, suppose the channel coefficients vary with time, i.e., we have a parallel Gaussian channel. At each time instant the players take turns to design the channel coefficients according to the rules described above. Corresponding to each sub-channel, player $2$ has control of one of the channel co-efficients. In this case, player $2$ can kill the degrees of freedom of the individual subchannels by using Theorem \ref{thm:main}. However, player $1$ still wins the game since $3/2$ degrees of freedom are achievable through the interference alignment scheme of \cite{cadambe_jafar:Kuserint} which codes across all parallel channels. Thus, in the time-varying case, player $1$ wins the game.

Note that, it is important that user $2$ has control of \emph{different} channel co-efficients over \emph{different} sub-channels. If user $2$ controls the same channel co-efficient of all the subchannels, it can, in fact, use Theorem \ref{thm:main} to kill the degrees of freedom of the channel. For example, consider the case where if user $2$ has control of all the entries of $\mathbf{H}_{1,2}$ over all subchannels, then it can choose $\mathbf{H}_{1,2} = \mathbf{H}_{1,3} \mathbf{H}_{2,2} (\mathbf{H}_{2,3})^{-1}$. 

\section{Multiple access outerbounds for the capacity of $3$ user interference channel}
\label{sec:macouterbound}
\begin{figure}[!tbp]
\begin{center}\input{MACouterbound.eepic}\end{center}
\caption{Multiple access outerbound for the classical $3$ user interference channel}
\label{fig:macouterbound}
\end{figure}
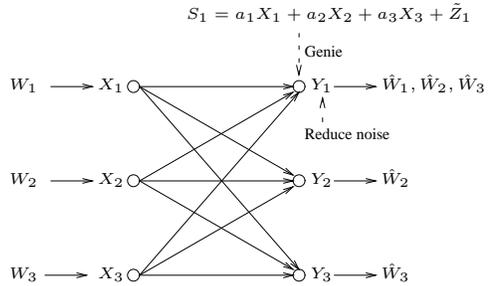
In this section, we provide an interesting application of the result of Theorem \ref{thm:main} in the form of a class of outerbounds for the classical (single-carrier) $3$ user interference  channel. The outerbound argument goes as follows. Consider any achievable coding scheme. Using this coding scheme, receiver $1$ can decode $W_1$. Our aim is to enhance receiver $1$ with enough information so that it can decode $W_2$ and $W_3$ as well (see Figure \ref{fig:macouterbound}). Then the capacity region of the multiple access channel(MAC) formed by the three transmitters and the (enhanced) receiver $1$ forms an outerbound for the capacity region of the interference channel. The improvements to receiver $1$ are described in the following steps
\begin{enumerate}
\item \emph{To help receiver $1$ decode $W_2$ :} Let a genie provide receiver $1$ with a $S_1 = a_1 X_1+a_2 X_2+a_3 X_3+\tilde{Z}_1$ where $\tilde{Z}_1$ is an AWGN term independent of $X_i,i=1,2,3$. Note that this side information effectively acts as an additional antenna at receiver $1$. The noise term $\tilde{Z}_1$ can possibly be correlated with other noise variables $Z_i,i=1,2,3$. Now, receiver $1$ can linearly combine its received signal with its side information to form $U_1 = \alpha Y_1 + \beta S_1$ to form another (noisy) linear combination of the codewords $X_i,i=1,2,3$. $\alpha$ and $\beta$ can be chosen such that the co-efficients of $X_1$ and $X_2$ in $U_1$ satisfy the conditions of Theorem \ref{thm:main}. Note that if these channel co-efficients already satisfy the condition of \ref{thm:main}, then side information of $S_1$ is not needed. Now, the proof of Theorem \ref{thm:main} implies that by sufficiently reducing the noise at receiver $1$, we can ensure that receiver $1$ decodes $W_2$ as well. Thus, with the aid of a genie and possibly reducing the noise, we have ensured that receiver $1$ can decode $W_2$. Note that neither the genie information, nor the reduction of noise reduce the capacity of this channel and therefore do not affect the outerbound argument.
\item \emph{To help receiver $1$ decode $W_3$ :} Receiver $1$, enhanced as described in the previous step, can now decode $W_1$ and $W_2$. We can now choose $\bar{\alpha}, \bar{\beta}, \bar{\gamma}$ such that \begin{eqnarray*}V_1&=&\bar{\alpha} X_1 + \bar{\beta} X_2 + \bar{\gamma} Y_1^{'} \\ &=& h_{3,1}X_1 + h_{3,2}X_2 + h_{3,3} X_3+\bar{\gamma}Z_1^{'}.\end{eqnarray*} Note that receiver $1$ can form $V_1$.  We use $Y_1^{'}$ and $Z_1^{'}$ above rather than $Y_1$ and $Z_1$ since the previous step involves reducing the noise at receiver $1$. Statistically, $V_1$ differs from $Y_3$ only in the variance of the noise term. Therefore, by further reducing the noise if required, receiver $1$ can also decode $W_3$. As in the previous step, it is important to note that the reduction of noise does not affect the outerbound argument
\end{enumerate}
Steps $1$ and $2$ above imply that the capacity region of the $3$ user Gaussian interference channel is outer-bounded by the capacity region of the single-input-multiple-output (SIMO) Gaussian MAC which receiver $S_1$ on one antenna and a reduced-noise version of $Y_1$ on the other. This class of bounds can be optimized over $a_i,i=1,2,3$ and the statistics of $\tilde{Z}_1$. Further, similar outerbounds can be found by enhancing receiver $2$ or receiver $3$  rather than receiver $1$. Note that, since a MAC with two antennas has $2$ degrees of freedom, this class of outerbounds is loose from the perspective of degrees of freedom. Using the Carliel's outerbounds on each of the two user channels contained within the $3$ user interference channel obtains a degrees of freedom outerbound of $3/2$ (See \cite{host-madsen_nosratinia:dofint,cadambe_jafar:Kuserint}). We now provide $2$ examples of this class of outerbounds.

\emph{Example 1:}
Here, we consider the interference channel formed on the first carrier of the parallel Gaussian interference channel described in Equations (\ref{eqn:parallel1})-(\ref{eqn:parallel2}) in the previous section.
In this channel, $h_{i,j} = 1, i \neq j, i,j \in \{1,2,3\}$. Also $h_{11}=h_{22}=1, h_{33} = -1$. With AWGN power at each receiver normalized to unity, the total transmit power at all the transmitters is defined as $\mbox{SNR}$. Consider any achievable coding scheme. Note that this channel already satisfies the conditions of Theorem \ref{thm:main}. Therefore, we do not need the aid of a genie. In fact, both receiver $1$ and receiver $2$ receive signals of the form $X_1+X_2+X_3+Z$ where $Z$ is an AWGN term of unit variance. Therefore, any message that can be decoded at receiver $2$ can also be decoded at receiver $1$ (and vice-versa). Receiver $1$ can hence decode $W_2$. Furthermore, receiver $1$ can compute $X_1+X_2-Y_1 = h_{3,1}X_1+h_{3,2}X_2+h_{3,3}X_3 - Z_1$. Since $(-Z_1)$ is a AWGN term having the variance as $Z_3$, receiver $1$ can decode $W_3$ without requiring any noise reduction. Thus, the capacity region of this channel is bounded by the capacity region of the multiple access channel formed at receiver $1$. The sum-capacity of this interference channel is therefore bounded by
$$ C_\Sigma \leq 1/2 \log\left(1+\mbox{SNR}\right)  $$
It can be easily verified that the sum-capacity of the interference channel corresponding to the second carrier of the parallel channel described by Equations (\ref{eqn:parallel1})-(\ref{eqn:parallel2}) can also be bounded as above.

\emph{Example 2:}
Consider the \emph{perfectly symmetric} $3$ user interference channel where $h_{i,i} = 1 \forall i=1,2,3$ and $h_{i,j} = h > 1, \forall i \neq j, i,j \in \{1,2,3\}$. Also, let the total transmit power be equal to $\mbox{SNR}$. Since the channel does not satisfy the conditions of Theorem \ref{thm:main}, a genie provides receiver $1$ with information of $S_1 = a_1 X_1+(1-h)X_2+X_3+\tilde{Z}_1$ where $\tilde{Z}_1$ is an i.i.d  AWGN term correlated with $Z_1$ such that $E\left[(Z_1+\tilde{Z}_1)^2)\right] = 1$. Note that since we started with an achievable coding scheme, receiver $1$ can decode $W_1$ using information from $Y_1$. Receiver $1$ can subtract the effect of $X_1$ from $S_1$ and $Y_1$ and to obtain 
$\tilde{S}_1 = (1-h)X_2 + \tilde{Z}_1$ and $\tilde{Y}_1 = h X_2 + hX_3 + Z_1$. Now receiver $1$ can now decode $X_2$ from $U_1hX_1+Y_1+S_1$ since it is of the form $hX_1+X_2+h X_3 + Z_2^{'}$ where $Z_2^{'}$ is a AWGN term with unit variance. Now that receiver $1$ is aware of $X_1$ and $X_2$, it can add appropriate terms to $Y_1$ to form $V_1 = h(h X_1 + hX_2 + X_3) + Z_1$. Since $h > 1$, $Y_3$ is a degraded version of $V_1$ which implies that receiver $1$ can decode $W_3$ as well.
Thus, all rates achievable in this interference channel, are achievable in the single-input-multiple-output (SIMO) multiple access channel with $3$ single antenna nodes respectively transmitting $X_1,X_2,X_3$ and a two-antenna node receiving $Y_1$ along the first antenna and $S_1$ along the second. Thus, the capacity region of this multiple access channel is an outer-bound for the capacity of the interference channel. Furthermore, parameters $a_1$ and $\tilde{Z}_1$ are parameters which can be used for optimization. So, for example, we can bound the sum-capacity $C_{\Sigma}(\mbox{SNR})$ of the $3$ user interference channel by 
$$ C_\Sigma(\mbox{SNR}) \leq \frac{1}{2}\min_{\begin{array}{c}(a_1, \tilde{Z}_1) \\E\left[(Z_1+\tilde{Z}_1)^2\right] \leq 1 \\ \tilde{Z}_1 \thicksim \mathcal{N}(0,\sigma^2)\end{array}} C_{\textrm{MAC}}(\rho,a_1,\tilde{Z}_1)$$
where
$$ C_{\textrm{MAC}}(\mbox{SNR},a_1,\tilde{Z}_1) = \log\left(\frac{|\mathbf{K_z}+\mathbf{\frac{\mbox{SNR}}{3} \mathbf{H}\mathbf{H}^\dagger}|}{|\mathbf{K_z}|}\right)$$
$|\mathbf{A}|$ indicates the detereminant of matrix $\mathbf{A}$, $\mathbf{K_z}$ indicates the covariance matrix corresponding to noise vector $[Z_1~~\tilde{Z}_1]^T$ and
$$\mathbf{H} = \left[ \begin{array}{ccc} 1& h& h \\ a_1 & (1-h) & 0 \end{array}\right]$$

\section{Conclusion}
We constructed a $3$ user interference channel with constant (i.e., not frequency-selective or time-varying) channel coefficients such that it has $1$ degree of freedom. Furthermore, we provided a mult-carrier extension of this channel such that separate coding over each carrier can only achieve sum rate $\log(\mbox{SNR})+o(\log(\mbox{SNR}))$ per carrier, while the actual capacity is  $3/2\log(\mbox{SNR})+o(\log(\mbox{SNR}))$ which can be achieved only through coding across carriers.  The result implies that, in general, independent coding over the various channel states of the parallel Gaussian interference channel is not capacity optimal. Thus, unlike parallel Gaussian point to point, multiple access and broadcast channels, parallel Gaussian interference channels are, in general, not separable.  The key is that even though interference alignment may not be possible over each carrier, it may still be accomplished by coding across carriers. 

An interesting question that remains open is the separability of parallel Gaussian interference channels for two users. The counterexample provided in this work applies to the $3$ user scenario and by simple extension to $K\geq 3$ users. However, since our examples rely on interference alignment which is only known to be relevant for interference channels with $3$ or more users, we have not shown that the $2$ user parallel Gaussian interference channel is inseparable. It is interesting to note that the $2$ user interference channel is separable under strong interference \cite{taek_chung:thesis}. 

The inseparability of the interference channel may have interesting implications, especially for the existence of single-letter capacity characterizations for interference channels. From a practical perspective, it prompts a closer look at the performance of separate encoding versus joint coding schemes in parallel Gaussian interference channels.

\bibliographystyle{ieeetr}
\bibliography{refs}
\end{document}

%% file: 3user.eepic
\setlength{\unitlength}{0.00041667in}
\begingroup\makeatletter\ifx\SetFigFont\undefined%
\gdef\SetFigFont#1#2#3#4#5{%
  \reset@font\fontsize{#1}{#2pt}%
  \fontfamily{#3}\fontseries{#4}\fontshape{#5}%
  \selectfont}%
\fi\endgroup%
{\renewcommand{\dashlinestretch}{30}
\begin{picture}(4725,2631)(0,-10)
\put(1575,2533){\ellipse{150}{150}}
\put(1575,1333){\ellipse{150}{150}}
\put(1575,133){\ellipse{150}{150}}
\put(3675,2533){\ellipse{150}{150}}
\put(3675,1333){\ellipse{150}{150}}
\put(3675,133){\ellipse{150}{150}}
\path(4125,2533)(4650,2533)
\path(4530.000,2503.000)(4650.000,2533.000)(4530.000,2563.000)
\path(4125,1333)(4650,1333)
\path(4530.000,1303.000)(4650.000,1333.000)(4530.000,1363.000)
\path(4125,133)(4650,133)
\path(4530.000,103.000)(4650.000,133.000)(4530.000,163.000)
\put(4725,2458){\makebox(0,0)[lb]{{\SetFigFont{6}{7.2}{\rmdefault}{\mddefault}{\updefault}$\hat{W}_1$}}}
\put(4725,1258){\makebox(0,0)[lb]{{\SetFigFont{6}{7.2}{\rmdefault}{\mddefault}{\updefault}$\hat{W}_2$}}}
\put(4725,58){\makebox(0,0)[lb]{{\SetFigFont{6}{7.2}{\rmdefault}{\mddefault}{\updefault}$\hat{W}_3$}}}
\path(1650,2533)(3600,2533)
\path(3480.000,2503.000)(3600.000,2533.000)(3480.000,2563.000)
\path(1650,1333)(3600,2458)
\path(3511.049,2372.048)(3600.000,2458.000)(3481.066,2424.019)
\path(1650,133)(3675,2458)
\path(3618.809,2347.807)(3675.000,2458.000)(3573.564,2387.214)
\path(1650,2533)(3600,1408)
\path(3481.066,1441.981)(3600.000,1408.000)(3511.049,1493.952)
\path(1650,1333)(3600,1333)
\path(3480.000,1303.000)(3600.000,1333.000)(3480.000,1363.000)
\path(1650,1333)(3600,208)
\path(3481.066,241.981)(3600.000,208.000)(3511.049,293.952)
\path(1650,133)(3600,133)
\path(3480.000,103.000)(3600.000,133.000)(3480.000,163.000)
\path(1650,133)(3600,1258)
\path(3511.049,1172.048)(3600.000,1258.000)(3481.066,1224.019)
\path(1650,2533)(3675,208)
\path(3573.564,278.786)(3675.000,208.000)(3618.809,318.193)
\path(525,2533)(1050,2533)
\path(930.000,2503.000)(1050.000,2533.000)(930.000,2563.000)
\path(525,1333)(1050,1333)
\path(930.000,1303.000)(1050.000,1333.000)(930.000,1363.000)
\path(450,133)(975,133)
\path(855.000,103.000)(975.000,133.000)(855.000,163.000)
\put(1125,2458){\makebox(0,0)[lb]{{\SetFigFont{6}{7.2}{\rmdefault}{\mddefault}{\updefault}$X_1$}}}
\put(1125,1258){\makebox(0,0)[lb]{{\SetFigFont{6}{7.2}{\rmdefault}{\mddefault}{\updefault}$X_2$}}}
\put(1125,58){\makebox(0,0)[lb]{{\SetFigFont{6}{7.2}{\rmdefault}{\mddefault}{\updefault}$X_3$}}}
\put(3825,2458){\makebox(0,0)[lb]{{\SetFigFont{6}{7.2}{\rmdefault}{\mddefault}{\updefault}$Y_1$}}}
\put(3825,1258){\makebox(0,0)[lb]{{\SetFigFont{6}{7.2}{\rmdefault}{\mddefault}{\updefault}$Y_2$}}}
\put(3825,58){\makebox(0,0)[lb]{{\SetFigFont{6}{7.2}{\rmdefault}{\mddefault}{\updefault}$Y_3$}}}
\put(0,2458){\makebox(0,0)[lb]{{\SetFigFont{6}{7.2}{\rmdefault}{\mddefault}{\updefault}$W_1$}}}
\put(0,1258){\makebox(0,0)[lb]{{\SetFigFont{6}{7.2}{\rmdefault}{\mddefault}{\updefault}$W_2$}}}
\put(0,58){\makebox(0,0)[lb]{{\SetFigFont{6}{7.2}{\rmdefault}{\mddefault}{\updefault}$W_3$}}}
\end{picture}
}

%% file: 3userconverse.eepic
\setlength{\unitlength}{0.00050000in}
\begingroup\makeatletter\ifx\SetFigFont\undefined%
\gdef\SetFigFont#1#2#3#4#5{%
  \reset@font\fontsize{#1}{#2pt}%
  \fontfamily{#3}\fontseries{#4}\fontshape{#5}%
  \selectfont}%
\fi\endgroup%
{\renewcommand{\dashlinestretch}{30}
\begin{picture}(10125,4333)(0,-10)
\put(900,3433){\ellipse{150}{150}}
\put(900,2233){\ellipse{150}{150}}
\put(900,1033){\ellipse{150}{150}}
\put(3000,3433){\ellipse{150}{150}}
\put(3000,2233){\ellipse{150}{150}}
\put(3000,1033){\ellipse{150}{150}}
\put(7350,3433){\ellipse{150}{150}}
\put(7350,2233){\ellipse{150}{150}}
\put(7350,1033){\ellipse{150}{150}}
\put(9450,3433){\ellipse{150}{150}}
\put(9450,1033){\ellipse{150}{150}}
\path(7421,3429)(9446,1104)
\path(9344.564,1174.786)(9446.000,1104.000)(9389.809,1214.193)
\path(7425,3433)(9375,3433)
\path(9255.000,3403.000)(9375.000,3433.000)(9255.000,3463.000)
\path(7425,2233)(9375,3358)
\path(9286.049,3272.048)(9375.000,3358.000)(9256.066,3324.019)
\path(7425,1033)(9450,3358)
\path(9393.809,3247.807)(9450.000,3358.000)(9348.564,3287.214)
\path(7425,2233)(9375,1108)
\path(9256.066,1141.981)(9375.000,1108.000)(9286.049,1193.952)
\path(9525,3433)(9900,3433)
\path(9780.000,3403.000)(9900.000,3433.000)(9780.000,3463.000)
\path(9525,1033)(9975,1033)
\path(9855.000,1003.000)(9975.000,1033.000)(9855.000,1063.000)
\path(7200,3733)(7500,3733)(7500,2008)
	(7200,2008)(7200,3733)
\path(7425,1033)(9375,1033)
\path(9255.000,1003.000)(9375.000,1033.000)(9255.000,1063.000)
\path(6900,3433)(7275,3433)
\path(7155.000,3403.000)(7275.000,3433.000)(7155.000,3463.000)
\path(6900,2233)(7275,2233)
\path(7155.000,2203.000)(7275.000,2233.000)(7155.000,2263.000)
\path(6900,1033)(7275,1033)
\path(7155.000,1003.000)(7275.000,1033.000)(7155.000,1063.000)
\put(10050,3358){\makebox(0,0)[lb]{{\SetFigFont{7}{8.4}{\rmdefault}{\mddefault}{\updefault}$\hat{W}_1,\hat{W}_2$}}}
\put(10125,958){\makebox(0,0)[lb]{{\SetFigFont{7}{8.4}{\rmdefault}{\mddefault}{\updefault}$\hat{W}_3$}}}
\put(6450,3358){\makebox(0,0)[lb]{{\SetFigFont{7}{8.4}{\rmdefault}{\mddefault}{\updefault}$W_1$}}}
\put(6450,2158){\makebox(0,0)[lb]{{\SetFigFont{7}{8.4}{\rmdefault}{\mddefault}{\updefault}$W_2$}}}
\put(6450,958){\makebox(0,0)[lb]{{\SetFigFont{7}{8.4}{\rmdefault}{\mddefault}{\updefault}$W_3$}}}
\path(975,3433)(2925,2308)
\path(2806.066,2341.981)(2925.000,2308.000)(2836.049,2393.952)
\path(975,1033)(2925,1033)
\path(2805.000,1003.000)(2925.000,1033.000)(2805.000,1063.000)
\path(975,3433)(3000,1108)
\path(2898.564,1178.786)(3000.000,1108.000)(2943.809,1218.193)
\dashline{60.000}(3000,4108)(3000,3583)
\path(2970.000,3703.000)(3000.000,3583.000)(3030.000,3703.000)
\dashline{60.000}(4050,1783)(3075,1783)(3075,2083)
\blacken\path(3105.000,1963.000)(3075.000,2083.000)(3045.000,1963.000)(3105.000,1963.000)
\path(975,3433)(2925,3433)
\path(2805.000,3403.000)(2925.000,3433.000)(2805.000,3463.000)
\path(975,2233)(2925,3358)
\path(2836.049,3272.048)(2925.000,3358.000)(2806.066,3324.019)
\path(975,1033)(3000,3358)
\path(2943.809,3247.807)(3000.000,3358.000)(2898.564,3287.214)
\path(975,2233)(2925,2233)
\path(2805.000,2203.000)(2925.000,2233.000)(2805.000,2263.000)
\path(975,2233)(2925,1108)
\path(2806.066,1141.981)(2925.000,1108.000)(2836.049,1193.952)
\path(975,1033)(2925,2158)
\path(2836.049,2072.048)(2925.000,2158.000)(2806.066,2124.019)
\path(3075,3433)(3450,3433)
\path(3330.000,3403.000)(3450.000,3433.000)(3330.000,3463.000)
\path(3075,2233)(3450,2233)
\path(3330.000,2203.000)(3450.000,2233.000)(3330.000,2263.000)
\path(3075,1033)(3450,1033)
\path(3330.000,1003.000)(3450.000,1033.000)(3330.000,1063.000)
\path(450,3433)(825,3433)
\path(705.000,3403.000)(825.000,3433.000)(705.000,3463.000)
\path(450,2233)(825,2233)
\path(705.000,2203.000)(825.000,2233.000)(705.000,2263.000)
\path(450,1033)(825,1033)
\path(705.000,1003.000)(825.000,1033.000)(705.000,1063.000)
\put(2475,4183){\makebox(0,0)[lb]{{\SetFigFont{7}{8.4}{\rmdefault}{\mddefault}{\updefault}Reduce Noise}}}
\put(4125,1708){\makebox(0,0)[lb]{{\SetFigFont{7}{8.4}{\rmdefault}{\mddefault}{\updefault}$X_1$}}}
\put(3375,1483){\makebox(0,0)[lb]{{\SetFigFont{7}{8.4}{\rmdefault}{\mddefault}{\updefault}Genie}}}
\put(3525,958){\makebox(0,0)[lb]{{\SetFigFont{7}{8.4}{\rmdefault}{\mddefault}{\updefault}$\hat{W}_3$}}}
\put(3525,2158){\makebox(0,0)[lb]{{\SetFigFont{7}{8.4}{\rmdefault}{\mddefault}{\updefault}$\hat{W}_2$}}}
\put(3525,3358){\makebox(0,0)[lb]{{\SetFigFont{7}{8.4}{\rmdefault}{\mddefault}{\updefault}$\hat{W}_1$}}}
\put(0,3358){\makebox(0,0)[lb]{{\SetFigFont{7}{8.4}{\rmdefault}{\mddefault}{\updefault}$W_1$}}}
\put(0,2158){\makebox(0,0)[lb]{{\SetFigFont{7}{8.4}{\rmdefault}{\mddefault}{\updefault}$W_2$}}}
\put(0,958){\makebox(0,0)[lb]{{\SetFigFont{7}{8.4}{\rmdefault}{\mddefault}{\updefault}$W_3$}}}
\put(1575,58){\makebox(0,0)[lb]{{\SetFigFont{8}{9.6}{\rmdefault}{\mddefault}{\updefault}(a)}}}
\put(8175,208){\makebox(0,0)[lb]{{\SetFigFont{8}{9.6}{\rmdefault}{\mddefault}{\updefault}(b)}}}
\end{picture}
}

%% file: MACouterbound.eepic
\setlength{\unitlength}{0.00041667in}
\begingroup\makeatletter\ifx\SetFigFont\undefined%
\gdef\SetFigFont#1#2#3#4#5{%
  \reset@font\fontsize{#1}{#2pt}%
  \fontfamily{#3}\fontseries{#4}\fontshape{#5}%
  \selectfont}%
\fi\endgroup%
{\renewcommand{\dashlinestretch}{30}
\begin{picture}(4796,3529)(0,-10)
\put(1575,2533){\ellipse{150}{150}}
\put(1575,1333){\ellipse{150}{150}}
\put(1575,133){\ellipse{150}{150}}
\put(3675,2533){\ellipse{150}{150}}
\put(3675,1333){\ellipse{150}{150}}
\put(3675,133){\ellipse{150}{150}}
\path(1650,2533)(3600,1408)
\path(3481.066,1441.981)(3600.000,1408.000)(3511.049,1493.952)
\path(1650,1333)(3600,1333)
\path(3480.000,1303.000)(3600.000,1333.000)(3480.000,1363.000)
\path(1650,1333)(3600,208)
\path(3481.066,241.981)(3600.000,208.000)(3511.049,293.952)
\path(1650,133)(3600,133)
\path(3480.000,103.000)(3600.000,133.000)(3480.000,163.000)
\path(1650,133)(3600,1258)
\path(3511.049,1172.048)(3600.000,1258.000)(3481.066,1224.019)
\path(1650,2533)(3600,2533)
\path(3480.000,2503.000)(3600.000,2533.000)(3480.000,2563.000)
\path(1650,1333)(3600,2458)
\path(3511.049,2372.048)(3600.000,2458.000)(3481.066,2424.019)
\path(1650,133)(3675,2458)
\path(3618.809,2347.807)(3675.000,2458.000)(3573.564,2387.214)
\path(1650,2533)(3675,208)
\path(3573.564,278.786)(3675.000,208.000)(3618.809,318.193)
\path(525,2533)(1050,2533)
\path(930.000,2503.000)(1050.000,2533.000)(930.000,2563.000)
\path(525,1333)(1050,1333)
\path(930.000,1303.000)(1050.000,1333.000)(930.000,1363.000)
\path(450,133)(975,133)
\path(855.000,103.000)(975.000,133.000)(855.000,163.000)
\path(4125,2533)(4650,2533)
\path(4530.000,2503.000)(4650.000,2533.000)(4530.000,2563.000)
\path(4125,1333)(4650,1333)
\path(4530.000,1303.000)(4650.000,1333.000)(4530.000,1363.000)
\path(4125,133)(4650,133)
\path(4530.000,103.000)(4650.000,133.000)(4530.000,163.000)
\dashline{60.000}(3675,3208)(3675,2683)
\path(3645.000,2803.000)(3675.000,2683.000)(3705.000,2803.000)
\dashline{60.000}(3975,2083)(3975,2383)
\path(4005.000,2263.000)(3975.000,2383.000)(3945.000,2263.000)
\put(3825,1258){\makebox(0,0)[lb]{{\SetFigFont{6}{7.2}{\rmdefault}{\mddefault}{\updefault}$Y_2$}}}
\put(3825,58){\makebox(0,0)[lb]{{\SetFigFont{6}{7.2}{\rmdefault}{\mddefault}{\updefault}$Y_3$}}}
\put(1125,2458){\makebox(0,0)[lb]{{\SetFigFont{6}{7.2}{\rmdefault}{\mddefault}{\updefault}$X_1$}}}
\put(1125,1258){\makebox(0,0)[lb]{{\SetFigFont{6}{7.2}{\rmdefault}{\mddefault}{\updefault}$X_2$}}}
\put(1125,58){\makebox(0,0)[lb]{{\SetFigFont{6}{7.2}{\rmdefault}{\mddefault}{\updefault}$X_3$}}}
\put(3825,2458){\makebox(0,0)[lb]{{\SetFigFont{6}{7.2}{\rmdefault}{\mddefault}{\updefault}$Y_1$}}}
\put(0,2458){\makebox(0,0)[lb]{{\SetFigFont{6}{7.2}{\rmdefault}{\mddefault}{\updefault}$W_1$}}}
\put(0,1258){\makebox(0,0)[lb]{{\SetFigFont{6}{7.2}{\rmdefault}{\mddefault}{\updefault}$W_2$}}}
\put(0,58){\makebox(0,0)[lb]{{\SetFigFont{6}{7.2}{\rmdefault}{\mddefault}{\updefault}$W_3$}}}
\put(4725,2458){\makebox(0,0)[lb]{{\SetFigFont{6}{7.2}{\rmdefault}{\mddefault}{\updefault}$\hat{W}_1,\hat{W}_2,\hat{W}_3$}}}
\put(4725,1258){\makebox(0,0)[lb]{{\SetFigFont{6}{7.2}{\rmdefault}{\mddefault}{\updefault}$\hat{W}_2$}}}
\put(4725,58){\makebox(0,0)[lb]{{\SetFigFont{6}{7.2}{\rmdefault}{\mddefault}{\updefault}$\hat{W}_3$}}}
\put(3750,2908){\makebox(0,0)[lb]{{\SetFigFont{6}{7.2}{\rmdefault}{\mddefault}{\updefault}Genie}}}
\put(2250,3358){\makebox(0,0)[lb]{{\SetFigFont{6}{7.2}{\rmdefault}{\mddefault}{\updefault}$S_1=a_1X_1+a_2X_2+a_3X_3+\tilde{Z}_1$}}}
\put(3750,1858){\makebox(0,0)[lb]{{\SetFigFont{6}{7.2}{\rmdefault}{\mddefault}{\updefault}Reduce noise}}}
\end{picture}
}